\newif\ifprocs
\procstrue
\procsfalse %last one wins (comment this line to get LIPIcs proceedings format)

\ifprocs
\documentclass[a4paper,USenglish,cleveref,autoref,thm-restate,numberwithinsect]{lipics-v2021}
\else
\documentclass{article}
\usepackage[utf8]{inputenc}
\usepackage{fullpage}
\fi

\usepackage{enumitem}
\usepackage{microtype}
\usepackage{lineno}

\usepackage{amssymb}
\usepackage{amsmath}
\usepackage{amsthm}
\usepackage{thmtools}
\usepackage{thm-restate}
\usepackage{mathtools}
\usepackage{csquotes}
\MakeOuterQuote{"}

\usepackage{xcolor}
\ifprocs
\hypersetup{colorlinks=true,allcolors=blue}
\usepackage{xspace}
\usepackage{xfrac}
\else
\usepackage[colorlinks=true,allcolors=blue]{hyperref}
\usepackage{xspace}
\usepackage{xfrac}
\fi

\usepackage{comment}
\usepackage{algorithm}
\usepackage{algorithmic} 
\usepackage{todonotes}

\ifprocs
\renewcommand{\paragraph}[1]{\subparagraph{#1}}
\newtheorem*{theorem*}{Theorem}
\newtheorem{question}{Question}
\newtheorem{fact}[theorem]{Fact}
\else
\newtheorem{theorem}{Theorem}[section]
\newtheorem*{theorem*}{Theorem}

\newtheorem{lemma}[theorem]{Lemma}
\newtheorem{claim}[theorem]{Claim}

\newtheorem{definition}{Definition}

\fi

\DeclareMathOperator{\poly}{poly}
\DeclareMathOperator{\polylog}{polylog}

\DeclareMathOperator{\lcs}{lcs}

\DeclareMathOperator*{\argmin}{arg\,min}

\newcommand{\dt}{\Delta}

\newcommand{\sym}{\mathcal{S}}

\renewcommand{\lcs}{\mathsf{LCS}}
\newcommand{\obj}{\texttt{Obj}}
\newcommand{\opt}{\texttt{OPT}}

\newcommand{\BI}{\textsc{BestFromInput}\xspace}
\newcommand{\BRO}{\textsc{MedianReconstruct}\xspace}
\newcommand{\FA}{\textsc{ApproxMedian}\xspace}
\newcommand{\FkA}{\textsc{Approx $k$-Median}\xspace}
\newcommand{\FkAS}{\textsc{Approx $k$-MedianStreaming}\xspace}

\providecommand{\eqdef}{\coloneqq}

\title{Clustering Permutations: New Techniques with Streaming Applications}

\ifprocs
\author{Diptarka Chakraborty}{National University of Singapore, Singapore}{diptarka@comp.nus.edu.sg}{} {Work partially supported by an MoE AcRF Tier 2 grant (WBS No.\ A-8000416-00-00) and an NUS ODPRT grant (WBS No.\ A-0008078-00-00).}%TODO mandatory, please use full name; only 1 author per \author macro; first two parameters are mandatory, other parameters can be empty. Please provide at least the name of the affiliation and the country. The full address is optional

\author{Debarati Das}{Pennsylvania State University, USA}{debaratix710@gmail.com}{}{}

\author{Robert Krauthgamer}{Weizmann Institute of Science, Israel}{robert.krauthgamer@weizmann.ac.il}{} {Work partially supported by ONR Award N00014-18-1-2364, the Israel Science Foundation grant \#1086/18, and a Minerva Foundation grant,
  and by the Israeli Council for Higher Education (CHE) via the Weizmann Data Science Research Center.}

\authorrunning{D.~Chakraborty, D.~Das and R.~Krauthgamer} %TODO mandatory. First: Use abbreviated first/middle names. Second (only in severe cases): Use first author plus 'et al.'

\else
\author{%
Diptarka Chakraborty%
\thanks{National University of Singapore.
  Work partially supported by an MoE AcRF Tier 2 grant (WBS No. A-8000416-00-00) and an NUS ODPRT grant (WBS No. A-0008078-00-00)
    Email: \texttt{diptarka@comp.nus.edu.sg}
  }
\and
Debarati Das%
\thanks{Pennsylvania State University.
        Email: \texttt{debaratix710@gmail.com}
}
\and
Robert Krauthgamer%
  \thanks{Weizmann Institute of Science.
    Work partially supported by ONR Award N00014-18-1-2364, the Israel Science Foundation grant \#1086/18, and a Minerva Foundation grant,
    and by the Israeli Council for Higher Education (CHE) via the Weizmann Data Science Research Center.
    Email: \texttt{robert.krauthgamer@weizmann.ac.il}
  }
}
\fi

\ifprocs
\Copyright{Debarati Das and Diptarka Chakraborty and Robert Krauthgame} %TODO mandatory, please use full first names. LIPIcs license is "CC-BY";  http://creativecommons.org/licenses/by/3.0/

 %TODO mandatory: Please choose ACM 2012 classifications from https://dl.acm.org/ccs/ccs_flat.cfm 

 %TODO mandatory; please add comma-separated list of keywords

\nolinenumbers %uncomment to disable line numbering

\EventEditors{Yael Tauman Kalai}
\EventNoEds{1}
\EventLongTitle{14th Innovations in Theoretical Computer Science Conference (ITCS 2023)}
\EventShortTitle{ITCS 2023}
\EventAcronym{ITCS}
\EventYear{2023}
\EventDate{January 10--13, 2023}
\EventLocation{MIT, Cambridge, Massachusetts, USA}
\EventLogo{}
\SeriesVolume{251}
\ArticleNo{82}
\else
\fi

\begin{document}
%\pagenumbering{gobble}
\maketitle

\begin{abstract}
We study the classical metric $k$-median clustering problem over a set of input rankings (i.e., permutations), which has myriad applications, from social-choice theory to web search and databases. 
A folklore algorithm provides a $2$-approximate solution in polynomial time
for all $k=O(1)$,
and works irrespective of the underlying distance measure, so long it is a metric;
however, going below the $2$-factor is a notorious challenge. 
We consider the Ulam distance, a variant of the well-known edit-distance metric, where strings are restricted to be permutations.
For this metric, Chakraborty, Das, and Krauthgamer [SODA, 2021]
provided a $(2-\delta)$-approximation algorithm for $k=1$, where $\delta\approx 2^{-40}$.
    
Our primary contribution is a new algorithmic framework for clustering a set of permutations.
Our first result is a $1.999$-approximation algorithm
for the metric $k$-median problem under the Ulam metric,
that runs in time $(k \log (nd))^{O(k)}n d^3$
for an input consisting of $n$ permutations over $[d]$.
In fact, our framework is powerful enough to extend this result
to the streaming model (where the $n$ input permutations arrive one by one)
using only polylogarithmic (in $n$) space.
Additionally, we show that similar results can be obtained even in the presence of outliers, which is presumably a more difficult problem.
\end{abstract}

\newpage
%\pagenumbering{arabic}

\section{Introduction}
\label{sec:intro}
Clustering is one of the ubiquitous tasks used in data analysis, which partitions a set of objects into several groups so that similar items lie in the same group. One of the most widely studied variants is the \emph{metric $k$-median clustering}.
In this problem, given an input set $S$ of $n$ data points, the goal is to find a set of $k$ median points from the underlying space (not necessarily from $S$) such that the sum of distances of all the data points in $S$ to its nearest median point is minimized. (See Section~\ref{sec:prelims} for a formal definition.) Throughout this paper, we consider the above variant and refer to it simply by the $k$-median problem. For $k=1$, the problem is also referred to as the \emph{geometric median} (or simply median) problem.
For many applications, it suffices to find an approximate solution to the problem, i.e., find a set of points from the metric space whose objective value approximates the minimum multiplicatively (for a formal definition, see Section~\ref{sec:prelims}).
This problem has been studied extensively in theory as well as in applied domains, both for $k=1$ and arbitrary $k$.
The complexity of the problem varies with the underlying metric space.
For $k=1$, perhaps the most well-studied version is over a Euclidean space (aka the Fermat-Weber problem),
for which a near-linear time $(1+\epsilon)$-approximation algorithm (for any $\epsilon > 0$) is known~\cite{cohen2016geometric}.
Other spaces that have been considered for the median problem include Hamming (folklore), the edit metric~\cite{Sankoff75, kruskal1983, NR03,chakraborty2021approximating}, rankings/permutations~\cite{DKNS01,ACN08,chakraborty2021approximating}, Jaccard distance between sets~\cite{CKPV10}, and many more~\cite{fletcher2008robust, minsker2015geometric, cardot2017online}.
The problem is clearly more challenging for general $k$.
For points in $\mathbb{R}^d$, Chen~\cite{Chen09} gave $(1+\epsilon)$-approximation algorithm with running time $O(ndk + 2^{(k/\epsilon)^{O(1)}}d^2 \log^{k+2}n)$ (see references therein for an overview) using \emph{coresets}.
For arbitrary metric spaces, $O(1)$-approximation algorithms (with a trade-off between the approximation factor and the running time) are known, e.g.~\cite{charikar1999constant, indyk1999sublinear, arya2001local, guha2003clustering, mettu2004optimal, Chen09}. Better approximation results are known for specific metric spaces, like Hamming~\cite{OR00}, shortest-path metric in a graph~\cite{thorup2005quick}, etc. 

One of the fundamental metrics (other than Euclidean and Hamming) that finds numerous applications is the \emph{edit metric}. The edit distance is a well-known dissimilarity measure between strings, which counts the minimum number of basic edit operations, like character insertion, deleting, and substitution, required to transform one string into the other. The $k$-median clustering problem over the edit metric is of utter importance in various domains, including computational biology~\cite{gusfield1997, pevzner2000computational}, DNA storage system~\cite{GBCDLSB13, RMRAJY17}, speech recognition~\cite{kohonen1985median}, and classification~\cite{martinez2000use}.
For $k=1$, it is also referred to as the \emph{median string} problem~\cite{kohonen1985median} (an equivalent formulation is known as \emph{multiple sequence alignment}~\cite{gusfield1997}). Despite being an important problem, we only know that it is NP-hard, even for $k=1$~\cite{HC00, NR03}. Although several heuristics exist~\cite{casacuberta1997greedy, kruzslicz1999improved, pedreira2007spatial, abreu2014new, hayashida2016integer, Mirabal19}, we do not know any approximation result better than what holds for arbitrary metrics.

We study the $k$-median problem over the \emph{Ulam metric},
which is a variant of the edit metric, by restricting strings to be permutations.
The Ulam metric of dimension $d$ is $(\sym_d,\dt)$,
where $\sym_d$ is the set of all the permutations over $[d]$
and $\dt(x,y)$ is the minimum number of character-move operations
needed to transform $x$ into $y$~\cite{AD99}.%
\footnote{In a permutation, a character move can be thought of as
  ``picking up'' a character and ``placing'' it in another position.
  Since it is equivalent to one deletion and one insertion,
  one can define the distance alternatively using character insertions and deletions~\cite{CMS01}.
}
Studying the Ulam metric is beneficial from two facets. First, it is a close variant of the edit metric defined over permutations and thus captures many inherent difficulties of the edit metric.
Thus, progress in the Ulam metric may provide insights to the same problem under a more general edit metric.
Second, it is a natural dissimilarity measure between rankings that arise in diverse areas ranging from social-choice theory~\cite{brandt2016handbook} to information retrieval~\cite{Harman92a},
and indeed has been studied from different algorithmic
perspectives~\cite{CMS01, CK06, AK10, AN10, NSS17, BS19}.

There is a folklore $O(n^{k+1}\cdot f(d))$-time algorithm 
(where $f(d)$ denotes the time to compute the distance between two points)
that provides a $2$-approximate solution to the $k$-median problem
for an arbitrary metric space,
by simply reporting the best $k$-tuple from the input set as the median points (see Procedure~\ref{alg:best-input}).
Breaking below the $2$-factor even in $n^{O(k)}$ time is one of the most notorious challenges, even for some specific metrics. So far, we do not know any affirmative result for the Ulam metric.
Very recently, Chakraborty, Das, and Krauthgamer~\cite{chakraborty2021approximating} provided a $(2-\delta)$-approximation algorithm only for the special case of $k=1$ (aka \emph{rank aggregation} under the Ulam),
where $\delta\approx 2^{-40}$ is a tiny constant.
On the contrary, the median problem concerning \emph{Kendall's tau distance},
yet another popular dissimilarity measure over permutations
(e.g.,~\cite{kemeny1959mathematics, young1988condorcet, young1978consistent, DKNS01, ACN08}),
possesses a PTAS~\cite{kenyon2007rank, Schudy2012thesis},
which can be extended to $k$-median using coresets~\cite{Danos21}. 

We fruther study this problem ($k$-median under the Ulam metric)
in the streaming model, i.e., when the input arrives sequentially;
more specifically, the $n$ input permutations $x_1,x_2,\ldots,x_n \in \sym_d$
arrive one by one.
This is sometimes called an insertion-only stream,
because an input permutation cannot be deleted after its arrival.
One of the challenges in this model is that the algorithm cannot freely access the input,
and the main goal is to devise an algorithm with an amount of space
that is sublinear (ideally, logarithmic) in $n$. 
The $k$-median problem in the insertion-only streaming model has also been studied extensively, e.g.,~\cite{guha2003clustering, Chen09, BravermanFLSZ21, braverman2021metric}.
However, no non-trivial result (other than what holds for an arbitrary metric)
is known for the Ulam metric.

\subsection{Our Contribution}
Our main result is a streaming algorithm for the $k$-median problem under the Ulam metric that achieves better than the 2-approximation factor.

\begin{restatable}{theorem}{mainthm}
\label{thm:main}
There is a (randomized) streaming algorithm that, given a set of permutations $x_1,x_2,\ldots,x_n \in \sym_d$ (arriving in streaming fashion), provides a 1.9999995-approximate solution to the (metric) $k$-median problem under the Ulam metric, using only $k^2 d \polylog(nd)$ bits of space, with high probability. Moreover, the algorithm has update time $(k\log n)^{O(1)} d \log^2 d$ and query time $(k\log (nd))^{O(k)}d^3$.
\end{restatable}
It is worth mentioning that here the input size is $O(nd\log d)$ bits (since each permutation of $\sym_d$ requires $O(d \log d)$ bits).
In our algorithm, we only need to maintain a subset of $k^2\polylog (nd)$ permutations. For $k=1$, we can improve the space-bound to only $O(d \log d \log^2 n)$ (\autoref{thm:1-median-stream}). Also, all our algorithms require only a single pass.

To achieve our result, we first develop a new algorithm framework that provides a 1.999-approximate solution to the median problem (i.e., for $k=1$) \emph{deterministically} (\autoref{thm:median}). Compared to the previous approximation result of~\cite{chakraborty2021approximating}, our algorithm is superior in various aspects. First, the approximation factor is 1.999, an improvement over the $2-\delta$, where $\delta\approx 2^{-40}$ in~\cite{chakraborty2021approximating}.
Second, our analysis is much simpler than that in~\cite{chakraborty2021approximating}, which was roughly divided into the following two cases:
In the first case, a large fraction of the optimum objective value is ``concentrated on a small fraction of symbols of $[d]$'',
and in the second is the optimum objective value is ``spread out throughout all the symbols of $[d]$'' (see also the Technical Overview below).
Analyzing these two cases separately and then combining them makes the entire analysis quite complicated and also affects the approximation guarantee.
In contrast, our analysis is pretty simple and argues that for $k=1$, 
\begin{itemize}
\item either there exists five input permutations from which we can reconstruct an approximate median (by essentially solving the \emph{feedback vertex set} problem over a special type of \emph{tournament graph});
\item or, there is an input permutation around which there is a large cluster (and thus would also provide a good approximate solution to the entire input);
\item or, there is an input permutation that is close to an (unknown) optimal median (and thus, again, would be a good approximate median).
\end{itemize}

Third, and perhaps most importantly, because the final output is
either an input permutation or derived from only five input permutations,
we essentially show that there are at most $n^5$ candidate points,
one of which provides a 1.999-approximate solution.
For general $k$, the same argument is clearly true also for each of the clusters
induced by an optimal $k$-median solution,
and thus we can easily extend our framework to the $k$-median problem, for arbitrary $k$, with running time $n^{O(k)}d^3$ (\autoref{thm:k-median}).
This running time can further be improved to $(k \log (nd))^{O(k)}n d^3$ using the sampling technique described in Section~\ref{sec:kstreaming}. Further, we can extend it to the $k$-median problem with \emph{outliers} (see Section~\ref{sec:outlier} for the definition and the details), presumably a more challenging problem. Such extensions were not possible to the algorithm of~\cite{chakraborty2021approximating}. Lastly, our algorithm can be implemented in a streaming fashion by storing (with the help of a "clever" sampling) only $k^2\poly \log (nd)$ input permutations and then running our offline algorithm (with a slight modification) on them. We emphasize that randomization is used only while implementing it in the streaming model.

\subsection{Technical Overview}

The algorithm of Chakraborty, Das and Krauthgamer~\cite{chakraborty2021approximating} computes,
given a set $S$ of $n$ permutations over $[d]$,
a $(2-\delta)$-approximate $1$-median under the Ulam metric for $\delta\approx 2^{-40}$.
Their algorithm starts with the following simple observation:
Fix an optimal median permutation $y^*$,
and let $\ell$ be the average distance between $y^*$ and the input permutations.
Now if $S$ contains a permutation $x$ such that the distance between $y$ and $x$ is $\le (1-\delta)\ell$, for some constant $\delta>0$,
then $x$ provides a $(2-\delta)$-approximation to the optimal objective value.
Otherwise, they first considered the case where the average distance is large, i.e., $\ell=\Omega(d)$.
In this case, using a counting argument, they show that there exists $x\in S$
such that at least $\Omega(n)$ other input permutations are at a distance at most $(2-\delta')\ell$. %The input set $S$ contains a permutation $x$ such that the set of unaligned characters in an optimal alignment between $x$ and $y$ has a large overlap with the set of unaligned characters in an optimal alignment between $x'$ and $y$ for $\Omega(n)$ different input permutations $x'\neq x$. Thus the distance between $x$ and $x'$ is strictly smaller than $2\ell$ and 
Thus taking $x$ as an approximate median provides better-than-$2$ approximation to the optimal objective value.
Of course, this $x\in S$ is not known, but outputting the input permutation that minimizes the objective serves the purpose.
In the other case, where the average distance $\ell$ is small,
if the cost is distributed only over a few symbols of $[d]$,
then restricting the input permutations only to these symbols gives rise to an instance with a large average objective,
and thus one can reuse the large-distance algorithm mentioned above.
If, however, the total cost is distributed over many symbols,
then for almost all the symbols in $[d]$, the cost is small.
Now consider two such small-cost symbols;
as both are aligned together in most of the input permutations,
their relative order in all these permutations is the same as that in $y^*$,
and thus can be computed by examining all the input permutations and taking the majority.
Moreover, using random sampling, this majority can be decided by looking at only $O(\log n)$ permutations.
However, this dependency on $\log n$ input permutations becomes ineffective
when we try to lift these ideas to $k$-median for $k>1$,
and this is where our new framework plays a crucial role. 

We start by providing an algorithm that shows that either an input permutation breaks the $2$-factor or the input set $S$ contains $5$ permutations from which we can derive a permutation whose distance is small from an optimal median $y^*$, and it thus gives better-than-$2$ approximation to the optimal objective value.
This dependency on the number inputs is significantly better than in~\cite{chakraborty2021approximating},
and it is specifically important for the $k$-median problem (general $k>1$),
where the grouping of the input permutations into $k$ clusters is not known. 
%As otherwise given a cluster $C_i$ (with optimal median $y_i$), we could retrieve the relative orders among symbols present in $y_i$, just by looking at the permutations in $C_i$. 
However, using our framework,
one can try all $\binom{n}{5}$ different ways of picking $5$ input permutations
and use them to derive candidates for approximate median.
These can serve as candidates for all the $k$ clusters (simultaneously)
without knowing the optimal partitioning of $S$.
This approach breaks down for the algorithm of~\cite{chakraborty2021approximating}, where the candidates are derived from $\Omega(\log n)$ inputs,
thus giving ${n} \choose {\log n}$ candidates overall.
We proceed to present next our new framework for computing $1$-median.

\paragraph{New algorithm for 1-median:}
Our algorithm (in Section~\ref{sec:median}) is based on a three-step framework. First, similar to~\cite{chakraborty2021approximating}, we use the fact that if the input set $S$ contains a permutation $x$ such that the distance between $x,y^*$ is at most $(1-\delta)\ell$ then $x$ serves as a $(2-\delta)$-approximate median. 

Next, for the sake of analysis, we fix an optimal alignment between each input $x$ and $y^*$. 
Let $I_x$ denote the symbols that are not aligned in this optimal alignment. In step two, we consider the scenario where all the permutations are at a distance at least $(1-\delta)\ell$ from $y^*$ and moreover, there is a subset $T \subseteq S$ containing five permutations $x_1,\dots, x_5$ such that for any pair $x_i, x_j$ their corresponding sets of unaligned characters have a very small overlap, i.e., $|I_{x_i}\cap I_{x_j}|\le \epsilon \ell$. In this case, we design an algorithm {\BRO} that just by using $x_1,\dots, x_5$ constructs a permutation $\tilde{x}$ such that the distance between $\tilde{x}$ and $y^*$ is at most $(1-\delta)\ell$. 
Let $B=\cup_{i,j\in [5]}(I_{x_i}\cap I_{x_j})$. Then $|B|\le 10 \epsilon \ell$. 
Now for any pair of symbols $a,b\in [d] \setminus B$, as each of $a,b$ can be unaligned in at most one $x_i$, together they are aligned in at least three out of five permutations $x_1, \dots, x_5$. 
Thus by looking at the relative order of $a$ and $b$ in these five permutations and taking the majority, we correctly deduce their order in $y^*$. 
This observation makes our algorithm framework significantly stronger than~\cite{chakraborty2021approximating} by reducing the dependency on the number of input permutations that are actually required to construct a good approximate median. 
However, we might not get the correct relative order for the pair of symbols that come from the set $B$. Instead, they may create conflict with other useful orders. To solve this, following a similar idea as in~\cite{chakraborty2021approximating}, we create a relative order graph $H$ that contains $d$ vertices corresponding to the $d$ symbols in the input permutations. For every pair of symbols $a,b$, we add an edge from $a$ to $b$ if at least in three of the five permutations from $T$, $a$ appears before $b$; otherwise, add an edge from $b$ to $a$. 
Next, we generate from $H$ an acyclic graph $H'$ by removing cycles iteratively while deleting the smallest one first. Here, since $H$ is a \emph{tournament}, the shortest cycle length is always three. Along with this, each cycle should contain at least one vertex from $B$. 
As $|B|\le 10\epsilon \ell$, this process removes very few vertices, 
and thus most of the symbols survive in $H'$. 
Next, we create a permutation $\tilde{x}$ by taking a topological ordering on $H'$ and appending the missing symbols at the end. 
Our analysis vastly differs from the one given in~\cite{chakraborty2021approximating}, where the relative order graph is not a tournament. This simplifies our cycle removal process and provides a better approximation guarantee. 

Lastly, we consider the case where there are at most four input permutations whose corresponding sets of unaligned symbols in the optimal alignment with $y$ have a small overlap. Again here, for at least one of them (call it $x$), there are $\Omega(n)$ other inputs $x'$ such that the sets of unaligned symbols for $x$ and $x'$ have a large overlap, and thus the distance between $x$ and $x'$ is strictly smaller than $2\ell$. 
Thus considering $x$ as an approximate median provides a better-than-$2$ approximation of the optimal objective value. 
Our algorithm also finds an input permutation that minimizes the total objective 
and finally outputs the best among those generated by {\BRO} and the best input permutation. 

\paragraph{Approximating $k$-median:}
Next, we show how our $1$-median algorithm can be extended to $k$-median for general $k$. 
For analysis purpose, fix $k$ optimal medians $y^*_1,\dots,y^*_k$,
and let $C_i$ be the set of input permutations served by $y^*_i$. 
Following the $1$-median algorithm, either $C_i$ contains a permutation that provides $<2$ approximation, or it includes $5$ permutations using which we can construct the approximate median for $C_i$. 
However, to start with, we do not know the optimal $k$-partition of the input set. 
Nevertheless, we can construct a set $M$ of potential approximate $k$ medians as follows: For each of the ${n} \choose {5}$ different choices of $5$ input permutations, create a permutation using the {\BRO} algorithm and add it to the set $M$. 
Next, add all input permutations to $M$.  
It is straightforward to see that the set $M$ contains $k$ permutations $\tilde{y}_1, \dots, \tilde{y}_k$ such that each $\tilde{y}_i$ is a better-than-$2$ approximate median for $C_i$. 
To identify these $k$ permutations, we try all possible size-$k$ subsets of $M$ and output the one that minimizes the total objective. 
We can bound the overall running time of the algorithm by $n^{O(k)}d^3$. 
The details of this algorithm can be found in Section~\ref{sec:k-median}.
This running time can further be improved to $(k \log (nd))^{O(k)}n^2d^3$ using the sampling technique used in the streaming algorithm described next.

\paragraph{Streaming Algorithm for Approximating k-Median:}
For the $k$-median problem, we have shown that for each cluster $C_i$, at most $5$ permutations are enough to construct an approximate median. However, without knowing the optimal $k$-partitioning of the input $S$, we need to try all possible ${n} \choose {5}$ choices. Unfortunately, this step becomes infeasible when the $n$ input permutations arrive in a stream, and we can afford to store only a few (preferably polylog) of them. Towards this, we design a streaming algorithm that requires several techniques, including an efficient sampling of the input permutations, \emph{coreset} construction, etc. This algorithm provides a $(2-\delta)$-approximation of $k$-median while storing $O(k^2\log^{4} n \log d)$ permutations. Next, we provide a brief overview of this algorithm. The details can be found in Section~\ref{sec:kstreaming}.

\vspace{2mm}
\noindent
\textbf{Sampling procedure.}
Let $\opt$ be the total optimal objective for all $k$ clusters. The key component of our streaming algorithm is a "clever" sampling technique that selects a set $R_1$ of $O(k\log^4 n \log d)$ permutations from the input stream and ensures that for any cluster $C_i$ if its optimal objective $O_i$ is at least a constant fraction of the average objective, i.e., $O_i=\Omega(\frac{\opt}{k})$, then either $R_1$ contains a permutation that serves as a better-than-$2$ approximate median for $C_i$, or there are $5$ permutations in $R_1$ such that applying the {\BRO} algorithm on them we can generate an approximate median. To explain this, we fix a cluster $C_i$ and let $n_i=|C_i|$, $y^*_i$ be an optimal median of $C_i$, and $\ell_i=O_i/n_i$ be the average objective. Next, we argue that if we sample each permutation in the input stream independently, uniformly at random with probability $\log^2 n/n_i$, then the sample set satisfies the above-mentioned properties. 

For this, we first consider the case when at least $1/\log n$ fraction of the permutations in $C_i$ are close, i.e., at a distance $<(1-\delta)\ell_i$ from the optimal median $y^*_i$. Note that all such permutations are good candidates for an approximate median. Also, following our sampling rate, we will sample at least one of them with a high probability.  

Thus from now on, we assume even fewer permutations are close to $y^*_i$. Note, there can be at most $\frac{n_i}{1+\delta}$ permutations in $C_i$ that are at distance $>(1+\delta)\ell_i$ from $y^*_i$. Hence a constant fraction of permutations is at a distance between $(1-\delta)\ell_i$ and $(1+\delta)\ell_i$ from $y^*_i$. Let $C_i^{avg}$ be the set of all these permutations. Now for every permutation $x\in C_i^{avg}$, we define a neighboring set $C(x)$ containing all other permutations $z\in C_i^{avg}$ such that the intersection between the set of unaligned characters in optimal alignments between $x, y^*_i$ and $z, y^*_i$ is large and thus they are close. 
Note here $x$ serves as a better-than-$2$ approximate median for $C(x)$.
Now we call $x$ to be in $C_i^{avg, dense}$ if $|C(x)|$ is large i.e., $|C(x)|=\Omega(|C_i^{avg}|)=\Omega(n_i)$. This means for each permutation $x$ in $C_i^{avg, dense}$ there is a large cluster of size $\Omega(n_i)$ around $x$ and as $x$ serves as a better-than-$2$ approximate median for $C(x)$, overall by considering $x$ to be the approximate median for $C_i$ we get $<2$ approximation of the optimal objective $O_i$.
Now, if $C_i^{avg, dense}$ is large, then again, by our sampling strategy with high probability, we will sample a permutation $x$ from it. 

Lastly, we consider the case where $C_i^{avg, dense}$ is small. Here using an iterative argument, we show that our sampling algorithm will sample at least $5$ permutations such that their corresponding set of unaligned symbols has a small overlap. Thus, using algorithm {\BRO}, we can construct an approximate median. To see this notice as $C_i^{avg, dense}$ is small with high probability, we will sample a permutation $x_1\in C_i^{avg}\setminus C_i^{avg, dense}$ and thus $C(x_1)$ is small. Moreover, as $C(x_1)$ is small, we will sample another permutation $x_2$ such that $C(x_2)$ is small and $x_2\notin C(x_1)$ and thus $x_1$ and $x_2$ have small overlap between their unaligned character set. We can continue this process and show overall, with a high probability, $5$ permutations can be sampled with the small overlap property. Thus we can apply algorithm \BRO on these $5$ strings and compute an approximate median with objective $<2O_i$.

Here for the sampling to work for cluster $C_i$, we use a sampling rate $\log^2 n/n_i$. However, as we do not know $n_i$ in advance, we try all sampling rates in $\{1/(1+\epsilon), 1/(1+\epsilon)^2,\dots, 1/n\}$ and we can ensure that one of them is arbitrarily close to $\log^2 n/n_i$. The challenge here is that for a sampling rate much larger than $\log^2 n/n_i$, the sample size can grow beyond our space limit. For this, whenever a sample set grows beyond $k\log^3 n$, we discard that set. Though this limits the space complexity, it can destroy a good sample set that is necessary to keep from computing an approximate median for $C_i$.
Next, we need to show for the right sampling rate, the sample set size always stays within the limit, and thus we do not miss the useful permutations. To ensure this in the uniform sampling process, we also incorporate a pruning strategy. 
Note with a sampling rate $\log^2 n/n_i$ (or $(1+\epsilon)\log^2 n/n_i$), with high probability, we sample at most $10 \log^2 n$ permutations from $C_i$. However, if there is a cluster $C_j$ such that $|C_j|>>|C_i|$, then we end up sampling much more permutations from $C_j$ and the sampling set size goes beyond $k\log^3 n$, and we discard it. To upper bound this while adding a new permutation to our sample set, we ensure all previously added permutations are at a distance at least $\beta \ell_i$ (for a small constant $\beta$) from it. Now to start with, as we assumed $O_i=\Omega(\frac{\opt}{k})$, we can show that the total number of sampled permutations from all clusters $C_j\neq C_i$ is bounded by $k\log^3 n$ with high probability. Also, because of the pruning, we get an extra additive error $\beta \ell_i$. This can be bounded by setting $\beta$ appropriately.

\vspace{2mm}
\noindent
\textbf{Approximating small objective clusters.}
Now we focus on the clusters whose objective is much smaller than the average objective $\opt/k$. Notice that the total objective contributed by the small objective clusters is very small. For these clusters, it suffices just to give a constant approximation of their optimal cost. For this we use the \emph{monotone faraway sampling} (MFS) from~\cite{braverman2021metric} to sample a set $R_2\subseteq S$ of size $O(k^2\log k \log(kn))$ such that for each cluster $C_i$, $R_2$ contains a permutation $\tilde{y}_i$ such that the objective of $C_i$ w.r.t. $\tilde{y}_i$ is at most $5O_i+\rho \opt/k$. Thus by keeping $\rho$ small, we achieve the required approximation. 

\vspace{2mm}
\noindent
\textbf{Computing approximate $k$-median using coreset.}
After sampling set $R_1$ and $R_2$, following the offline $k$-median algorithm, we design a potential $k$-median set $\tilde{M}$ by adding all permutations of $R_1, R_2$ to $\tilde{M}$. Together with this for each subset $T$ of $R_1$ of size $5$, run algorithm $\BRO(T)$ and add the output $\tilde{x}_T$ to $\tilde{M}$. However, to decide which $k$ permutations from $\tilde{M}$ minimize the total objective, we need access to the actual input set $S$ (which we cannot store using a small space). Now again, to upper bound the space complexity, instead of storing $S$ explicitly, we construct a $(k,\lambda)$-coreset (see Section~\ref{sec:kstreaming} for the definition) for $S$ with respect to the implicit potential median set $M$.
Here $M$ is a set containing all permutations from $S$. 
Moreover, for each subset $T\subseteq S$ of size $5$, the output of $\BRO(T)$ is also present in $M$.
This coreset ensures that 
considering any $k$-size subset of $\tilde{M}$ as a potential $k$-median if we compute
the corresponding objectives for set $S$ and the coreset then these two objectives are close when $\lambda$ is small. Thus we can use the coreset to decide the best candidate $k$-median from $\tilde{M}$.
Moreover, as $|M| \leq O(n^5)$, following~\autoref{thm:coreset-streaming}, we can show that the coreset size is also bounded by $O(k^2\log^2 n)$. We remark that to make our algorithm space efficient, instead of explicitly storing each $\tilde{x}_T$ to $\tilde{M}$, we recompute it whenever $\tilde{x}_T$ is a part of a candidate $k$-median.

\paragraph{Approximating $k$-median with outliers:}
We further extend our $k$-median algorithm in the presence of \emph{outliers}. In the $k$-median with outliers problem, given a parameter $p \in [0,1)$ and input set $S$, the goal is to find a set of at most $k$ permutations that minimizes the $k$-median objective value of a subset of $S$ of size at least $(1-p)|S|$. (See Section~\ref{sec:outlier} for a formal definition.) Using the argument the same as that for the $k$-median problem (without outlier), we claim that a subset of at most $k$ permutations from a candidate set $M$ (of potential medians) of size at most $O(n^5)$ achieves a 1.999-factor approximation for the outlier variant as well (\autoref{thm:k-median-outlier}).

\subsection{Conclusion}
In this paper, we study the (metric) $k$-median problem under the Ulam metric, which is known to be $\mathsf{NP}$-hard even for $k=2$. The Ulam metric is of utter importance because it is a variant of the more general edit metric and an interesting dissimilarity measure over rankings (or permutations). There is a folklore 2-approximation algorithm that works for any metric space, and breaking this factor-$2$ barrier is one of the interesting challenges. Despite being an important metric, the problem under the Ulam metric does not possess any better approximation algorithm. For the special case of $k=1$, there is a $(2-\delta)$-approximation (where $\delta\approx2^{-40}$) algorithm known~\cite{chakraborty2021approximating}. However, that algorithm does not provide any non-trivial result for arbitrary values of $k$.

We provide a 1.9999995-approximation algorithm for the $k$-median problem, with running time $(k \log (nd))^{O(k)}n d^3$. Moreover, our algorithm works in the insertion-only streaming model, using only polylogarithmic (in the number of input permutations) space. Further, we can extend our framework to get a similar result even in the presence of outliers, which is presumably a more complex problem. We also would like to highlight that our framework is not very specific to the Ulam metric; in fact, it (with slight modification) also provides a similar result for other metrics like Kendall's tau defined over rankings/permutations. One exciting direction is to see whether our new framework can give similar results to the other known distance measures involving rankings, such as Spearman's footrule, Minkowski, Cayley, swap-and-mismatch, etc. Another stimulating future direction is to use this new framework to get a similar result for another essential variant of the clustering problem, namely the $k$-center clustering problem (under the Ulam metric). Finally, extending our result to the more general edit metric (even with certain restrictions on the input) would, of course, be super intriguing.

\section{Preliminaries}
\label{sec:prelims}
\paragraph{Notations.}We use $[d]$ to denote the set $\{1,2,\cdots,d\}$. Let $\sym_d$ denote the set of all permutations over $[d]$. Throughout the paper, we consider a permutation $x$ (over $[d]$) as a sequence $a_1,a_2,\cdots,a_d$ such that $x(i)=a_i$.

\paragraph{Ulam metric and the $k$-median problem.}For any two permutations $x,y \in \sym_d$, the \emph{Ulam distance} between them, denoted by $\dt(x,y)$, is the minimum number of character move operations\footnote{A single move operation in a permutation can be thought of as ``picking up'' a character from its position and then ``inserting'' that character in a different position.} that is needed to transform $x$ into $y$. Equivalently, it can be defined as $d-|\lcs(x,y)|$, where $\lcs(x,y)$ denotes a \emph{longest common subsequence} between $x$ and $y$.

For any two permutations (permutations) $x$ and $y$ of lengths $d_x$ and $d_y$ respectively, an \emph{alignment} $g$ is a function that maps $[d_x]$ to $[d_y]\cup \{\bot\}$ such that:
\begin{itemize}
\item $\forall i\in[d_x],\text{if } g(i)\neq \bot, \text{ then } x(i)=y(g(i)) $;
\item For any two distinct $i,j \in [d_x]$ where $g(i), g(j)\neq\bot$, $i<j\Leftrightarrow g(i)<g(j)$.
\end{itemize}

For an alignment $g$ between two permutations (permutations) $x$ and $y$, we say $g$ \emph{aligns} a character $x(i)$ with some character $y(j)$ if and only if $j=g(i)$. Thus the alignment $g$ is essentially a common subsequence between $x$ and $y$.

For any permutation $x \in \sym_d$ and a set $Y\subseteq \sym_d$, let us define the distance between $x$ and $Y$ as $\dt(x,Y):=\min_{y \in Y}\dt(x,y)$ (i.e., the minimum distance between $x$ and a permutation from $Y$). Given a set $S \subseteq \sym_d$ and a subset $Y \subseteq \sym_d$, we define the \emph{median objective value} of $S$ with respect to $Y$ as $\obj(S,Y):=\sum_{x \in S}\dt(x,Y)$. 

Given a set $S\subseteq \sym_d$, the \emph{$k$-median} problem asks to find a subset $Y \subseteq \sym_d$ of size at most\footnote{Here $Y$ is not a multi-set. If we allow $Y$ to be a multi-set, then we can ask $Y$ to be of size exactly $k$. Note, both formulations are equivalent.} $k$ such that $\obj(S,Y)$ is minimized, i.e., $Y^*=\argmin_{Y \subseteq \sym_d: |Y| \le k} \obj(S,Y)$. We refer to the set $Y^*$ as $k$-median of $S$. We refer $\obj(S,Y^*)$ as $\opt(S)$, or simply $\opt$ when $S$ is clear from the context. Note, for $k=1$, $y^*=\argmin_{y\in \sym_d} \obj(S,y)$ is referred to as a \emph{median} (or geometric median or 1-median). We call a set $\tilde{Y}$ a $c$-approximate $k$-median of $S$, for some $c>0$, if $\obj(S,\tilde{Y}) \le c \cdot \opt(S)$. Further, note, each set $\{y_1,\cdots,y_k\}$ induces a partitioning (clustering) of $S$ into $k$-clusters $C_1,\cdots, C_k$, where $C_i:=\{x \in S \mid \dt(x,y_i) \le \dt(x,y_j) \text{ for all }j \ne i\}$ (if for some $x\in S$, $\dt(x,y_i)=\dt(x,y_j)$ for some $i\ne j$, then break the ties arbitrarily to form $C_i$'s).

It is worth emphasizing that in the above definition of the $k$-median problem, the $k$-median set $Y^*$ need not be a subset of the input $S$. In the literature, this variant is sometimes referred to as the continuous $k$-median problem. On the other hand, the discrete variant asks to find a set $Y^*$ of size at most $k$, strictly from $S$ that minimizes the median objective value (over all the subset of $S$ of size at most $k$). It follows directly from the triangle inequality that any optimum discrete $k$-median set is a 2-approximate solution to the (continuous) $k$-median problem.

Since in the discrete version, the median points are necessarily from $S$, by brute force over all the $O(n^k)$ (where $|S|=n$) possibilities, we can compute an optimum solution. We refer to this algorithm as Procedure {\BI} (Procedure~\ref{alg:best-input}). So, the Procedure {\BI} provides a 2-approximate solution to the (continuous) $k$-median problem. The running time is $O(n^k + n^2 d\log d)$, since for any $x,y \in \sym_d$, we can compute $\dt(x,y)$ in $O(d \log d)$ time.

\begin{algorithm}
	\begin{algorithmic}[1]
		\REQUIRE $S\subseteq \sym_d$.
		
		\ENSURE A subset $Y \subseteq S$ of size at most $k$.
		
		% \vspace{1mm}
		% \hrule\vspace{1mm}
		
		\STATE For all pairs of permutations $x_i,x_j \in S$, compute $\dt(x_i,x_j)$
		
		\RETURN $\argmin_{Y \subseteq S: |Y|\le k} \obj(S,Y)$.

		\caption{{\BI}$(S,k)$}
		\label{alg:best-input}
	\end{algorithmic}
\end{algorithm}

\section{Approximation Algorithm for $1$-Median}
\label{sec:median}

\begin{theorem}
\label{thm:median}
There is a deterministic polynomial-time algorithm that, given a set $S$ of $n$ permutations over $[d]$, finds a $1.999$-approximate median.
\end{theorem}

\paragraph{Description of the algorithm.}Let us start with the description of our algorithm. Our algorithm consists of two procedures. The first procedure is {\BI} (by setting $k=1$), which simply outputs a permutation $\tilde{y} \in S$ with the minimum median objective value among all the inputs, i.e., $\tilde{y} = \argmin_{y \in S} \sum_{x \in S} \dt(x,y)$. The second procedure enumerates all subsets of $S$ of size five (i.e., 5-tuples of input permutations) and runs the procedure {\BRO}. For a subset $T \subseteq S$, {\BRO} works as follows: It constructs a directed graph $H$ with vertex set $[d]$ and edge set 
\[
E(H):=\{(a,b)\mid \text{$a$ appears before $b$ in at least three permutations of $T$}\}.
\]
Observe, the graph $H$ is a tournament\footnote{A directed graph is called a \emph{tournament} if between every pair of vertices there is a directed edge.}, but may not be acyclic. Next, the procedure iterates over all the vertices and while iterating over a vertex $v$, it finds a shortest cycle containing $v$ and deletes all its vertices (along with all the incident edges). Let $H'$ be the final resulting acyclic graph. Then the procedure performs a topological sorting on the vertices of $H'$ and let $\tilde{x}'$ denote the sorted ordering. Finally, it appends the remaining symbols ($[d]\setminus V(H')$) at the end of $\tilde{x}'$ in an arbitrary order and outputs the resulting permutation $\tilde{x}$.

\begin{algorithm}
	\begin{algorithmic}[1]
		\REQUIRE $T \subseteq S$.
		
		\ENSURE A permutation $\tilde{x}$ over $[d]$.
		
		% \vspace{1mm}
		% \hrule\vspace{1mm}
		
		\STATE $H\gets ([d],E)$ where\\
                $\qquad\quad E = \{(a,b) \mid \text{$a$ appears before $b$ in at least three permutations of $T$}\}  $

		\FORALL{$v \in [d]$}
		
		\STATE $\mathcal{C}_{\min}\gets $ {cycle of minimum length containing $v$ in $H$}  
		
		\STATE $H=H - V(\mathcal{C}_{\min})$ 
		\ENDFOR
		
        \STATE ${H'} \gets H$
		
		\STATE $\tilde{x}' \gets $ permutation formed by topological ordering of $V(H')$
		
		\STATE $\tilde{x} \gets $ permutation formed by appending to  $\tilde{x}'$ the symbols $[d]\setminus V(H')$ in an arbitrary order
		
		\RETURN $\tilde{x}$.

		\caption{{\BRO}$(T)$}
		\label{alg:ptwo}
	\end{algorithmic}
\end{algorithm}

For a subset $T\subseteq S$, let us denote the output of {\BRO} by $\tilde{x}_T$. Consider the set 
\[
M=\{\tilde{y}\} \cup \{\tilde{x}_T \mid \text{ for all }T\subseteq S \text{ such that }|T|=5\}.
\]
The final algorithm {\FA}($S$) (Algorithm~\ref{alg:final}) outputs the best permutation $z$ among the set $M$ that minimizes the median objective value, i.e., $z=\argmin_{y \in M} \sum_{x \in S} \dt(x,y)$.

\begin{algorithm}
	\begin{algorithmic}[1]
		\REQUIRE $S\subseteq \sym_d$.
		
		\ENSURE A subset $Y \subseteq S$ of size at most $k$.
		
		% \vspace{1mm}
		% \hrule\vspace{1mm}
		\STATE Initialize an empty set $M$
		
		\STATE $\tilde{y} \leftarrow {\BI}(S,1)$
		
		\STATE Add $\tilde{y}$ to $M$
		
		\STATE For all the subsets $T \subseteq S$ of size 5, run {\BRO}(T) and add the output to $M$
		
		\RETURN $\argmin_{y \in M} \obj(S,y)$.

		\caption{{\FA}$(S)$}
		\label{alg:final}
	\end{algorithmic}
\end{algorithm}

\paragraph{Running time analysis.}Note, each $\dt(x,y)$ computation takes $O(d \log d)$ time. Then the first procedure {\BI} takes only $O(n^2 d \log d)$ time. There are at most $O(n^5)$ subsets of $S$ of size exactly five. For each such subset, the {\BRO} procedure takes $O(d^2)$ time to construct the graph $H$. Then, computing a minimum length cycle passing through a vertex $v$ at each iteration takes $O(d^2)$ time. Since it iterates over all the vertices $v \in [d]$, the running time for the whole cycle removal step is $O(d^3)$. The topological ordering can be performed in $O(d^2)$ time. So the running time of {\BRO} is $O(d^3)$. Hence, the overall running time of the final algorithm {\FA} is $O(n^5 d^3)$. Later in Section~\ref{sec:kstreaming}, we will comment on how to reduce the running time to $\tilde{O}(d^3)$.

\paragraph{Analyzing the approximation factor.}Suppose $S=\{x_1,x_2,\cdots,x_n\}$. Let $x^*$ be an arbitrary optimal median of $S$. So, $\opt(S)=\sum_{x_i \in S}\dt(x_i,x^*)$. For each $x_i \in S$, consider an arbitrary optimal alignment between $x_i$ and $x^*$, and let $I_{x_i}$ (or for brevity, $I_i$) denote the set of unaligned symbols ($\subseteq [d]$) with respect to this alignment. Recall, by the definition, $\dt(x_i,x^*) = |I_i|$ for each $x_i \in S$. WLOG assume, $|I_1| \le |I_2| \le \cdots \le |I_n|$.

Consider $\epsilon=0.03319$ and $\alpha=\epsilon/11$. WLOG assume,
\begin{equation}
\label{eq:I1-assumption}
    |I_1| \ge  (1-\alpha) \opt/n.
\end{equation}
Otherwise, 
\begin{align}
    \label{eq:I1-noassumption}
    \obj(S,x_1) &\le \sum_{x_i \in S}\dt(x_i,x_1) \nonumber\\
    &\le \sum_{x_i \in S}(\dt(x_i,x^*) + \dt(x^*,x_1))&&\text{(by triangle inequality)}\nonumber\\
    &\le (2-\alpha) \opt.
\end{align}
It is straightforward to see that $\obj(S,\tilde{y}) \le \obj(S,x_1)$, and thus the final output $z$ satisfies $\obj(S,z) \le (2-\alpha)\opt$. So from now, we assume~\autoref{eq:I1-assumption}.

\begin{lemma}
\label{lem:two-cases}
Consider $\epsilon=0.03319$ and $\alpha=\epsilon/11$. Then one of the following holds:
\begin{enumerate}
    \item \label{itm:5-tuple}Either there are five inputs $x_{i_1}, x_{i_2},x_{i_3},x_{i_4},x_{i_5}$ (with $i_1<\cdots<i_5$) such that for any two $r,\ell \in \{i_1,i_2,i_3,i_4,i_5\}$ with $\ell > r$, $|I_r \cap I_\ell| \le \epsilon |I_r|$ and $|I_{i_4}| \le (1+\alpha)\opt/n$;
    \item \label{itm:bestinput}Or there exists $x_j \in S$ such that $\obj(S,x_j) \le 1.999\cdot \opt$.
\end{enumerate}
\end{lemma}
\begin{proof}
Let us consider the set of \emph{far} points, defined as 
\[
F:=\{x_i \in S \mid |I_i| \ge (1+\alpha)\opt/n\}.
\]
Let $\bar{F}:=S\setminus F$. For any subset $R\subseteq S$, let us define $\opt_R := \sum_{x_i \in R} \dt(x_i,x^*)$. 

Recall, we assume that $|I_1|\le \cdots \le|I_n|$. Then it is straightforward to see that for all $x_i \in S$, $\dt(x_i,x_1) \le \dt(x_i,x^*) + \dt(x^*,x_1) = |I_i| + |I_1|\le 2|I_i| = 2\dt(x_i,x^*)$. Further, observe, by an averaging argument, $|I_1| \le \opt/n$. Then for all $x_i \in F$,
\begin{align*}
    \dt(x_i,x_1) &\le |I_i| + |I_1|&&\text{(by the triangle inequality)}\\
    &\le |I_i| + \opt/n\\
    &\le |I_i| + |I_i|/(1+\alpha) \le (2-\alpha/2) |I_i|.
\end{align*}
As a consequence, we get that
\begin{equation}
    \label{eq:I1-good}
    \obj(S,x_1) \le 2 \opt_{\bar{F}} + (2-\alpha/2)\opt_F = 2\opt - \frac{\alpha}{2}\opt_F.
\end{equation}
So if $\opt_F \ge \frac{2}{3} \opt$, we get that $\obj(S,x_1) \le 1.999\cdot \opt$. So from now, assume
\begin{equation}
    \label{eq:far-light}
    \opt_F < \frac{2}{3}\opt.
\end{equation}

\iffalse
It is immediate from our assumption~\autoref{eq:I1-assumption} that $|F| \le |\bar{F}|$, and thus $|\bar{F}| \ge n/2$. Hence,
\begin{equation}
    \label{eq:opt-close}
    \opt_{\bar{F}}:=\sum_{x_i \in \bar{F}} \dt(x_i , x^*)\ge \frac{n}{2}\cdot (1-\alpha)\opt/n=\frac{1-\alpha}{2} \opt.
\end{equation}

\fi

Next, consider the following procedure $\mathcal{A}$ that processes $x_1,\cdots,x_n \in S$ one by one. Initialize a set $T\leftarrow {x_1}$. For each $x_i\in S$, if for all $x_j\in T$, $|I_j \cap I_i| \le \epsilon |I_j|$, add $x_i$ in $T$. Break when $|T|\ge 5$.

It is worth noting that the above procedure is considered only for the sake of analysis. Now when the above procedure terminates, suppose $T=\{x_{i_1},x_{i_2},\cdots,x_{i_5}\}$, where $1=i_1<i_2<\cdots<i_5$, and $x_{i_4} \in \bar{F}$. Then clearly it satisfies Item~\ref{itm:5-tuple} of the statement of the lemma.

If not, then either $x_{i_4} \in F$ or $|T|\le 4$. By the procedure $\mathcal{A}$, $x_{i_4} \in F$ implies that for all $x_i \in \bar{F}$, there exists $x_j \in T \cap \bar{F}$ such that $|I_j \cap I_i| > \epsilon |I_j|$. Then by a simple averaging argument, there exists $j \in T \cap \bar{F}$ and $R\subseteq \bar{F}$ such that 
\begin{itemize}
    \item $\opt_R \ge \frac{\opt_{\bar{F}}}{|T\cap \bar{F}|} \ge \frac{\opt_{\bar{F}}}{3}$; and
    \item For all $x_i \in R$, $|I_j \cap I_i|>\epsilon |I_j|$.
\end{itemize}
Consider this $j \in T \cap \bar{F}$ and $R\subseteq \bar{F}$. It follows from the triangle inequality that 
\begin{enumerate}
    \item[(i)] For all $x_i \in F$, $\dt(x_i,x_j) \le 2 \dt(x_i,x^*)$;
    \item[(ii)] For all $x_i \in \bar{F}$, $\dt(x_i,x_j) \le (2+3\alpha)\dt(x_i,x^*)$;
    \item[(ii)] For all $x_i \in R$, $\dt(x_i,x_j) \le (2-\epsilon)\dt(x_i,x^*)$.
\end{enumerate}
To see this, observe, for any two $x_i,x_r \in S$, $\dt(x_i,x_r) \le |I_i| + |I_r| - |I_i\cap I_r|$. Now, since for all $x_i \in F$, $|I_j| \le |I_i|$, the first item follows. For the second item, observe, for all $x_i \in \bar{F}$, $|I_j| \le (1+\alpha)\opt/n \le (1+3\alpha)|I_i|$ (since $|I_i| \ge (1-\alpha)\opt/n$ by assumption~\autoref{eq:I1-assumption}). For the third item, note, for all $x_i \in R$, $|I_i \cap I_j| > \epsilon |I_j|$, and further $|I_j| \le |I_i|$ (by the description of procedure $\mathcal{A}$).

Thus
\begin{align*}
    \obj(S,x_j) &= (2-\epsilon)\opt_R + 2 \opt_F + (2+3\alpha)\opt_{\bar{F}\setminus R}\\
    &\le 2 \opt - (\frac{\epsilon + 3\alpha}{3} - 3\alpha) \opt_{\bar{F}}&&\text{(since $\opt_R \ge \frac{\opt_{\bar{F}}}{3}$)}\\
    &\le (2-5\alpha/9)\opt &&\text{(by~\autoref{eq:far-light} and $\epsilon=11\alpha$)}\\
    &\le 1.999 \cdot \opt &&\text{(for $\alpha=\epsilon/11=0.03319/11$)}.
\end{align*}
When $x_{i_4}\in \bar{F}$, but $|T| \le 4$, in a similar way we can argue that there exists $x_j \in T$ and $R\subseteq S$ such that
\begin{itemize}
    \item $\opt_R \ge \frac{\opt}{4}$; and
    \item For all $x_i \in R$, $|I_j \cap I_i|>\epsilon |I_j|$.
\end{itemize}
Hence, again, we can argue as before that
\begin{enumerate}
    \item[(i)] For all $x_i \in F$, $\dt(x_i,x_j) \le 2 \dt(x_i,x^*)$;
    \item[(ii)] For all $x_i \in \bar{F}$, $\dt(x_i,x_j) \le (2+3\alpha)\dt(x_i,x^*)$;
    \item[(ii)] For all $x_i \in R$, $\dt(x_i,x_j) \le (2-\epsilon)\dt(x_i,x^*)$.
\end{enumerate}
Thus
\begin{align*}
    \obj(S,x_j) &= (2-\epsilon)\opt_R + 2 \opt_{F\setminus R} + (2+3\alpha)\opt_{\bar{F}\setminus R}\\
    &=2 \opt -\epsilon \opt_R + 3 \alpha \opt_{\bar{F}\setminus R}\\
    &\le 2 \opt -\epsilon \opt_R + 3 \alpha (\opt-\opt_R)\\
    &= 2 \opt - (\frac{\epsilon + 3\alpha}{4} - 3\alpha) \opt&&\text{(since $\opt_R \ge \frac{\opt}{4}$)}\\
    &\le (2-\alpha/2)\opt &&\text{(by setting $\epsilon=11\alpha$)}\\
    &\le 1.999 \cdot \opt &&\text{(for $\alpha=\epsilon/11=0.03319/11$)}.
\end{align*}
This concludes the proof.
\end{proof}

Clearly, if there exists $x_j \in S$ such that $\obj(S,x_j) \le 1.999\cdot \opt$, then 
\[
\obj(S,z) \le \obj(S,\tilde{y}) \le \obj(S,x_j) \le 1.999\cdot \opt.
\]
So it only remains to show that if there are five inputs $x_{i_1}, x_{i_2},x_{i_3},x_{i_4},x_{i_5}$ (with $i_1<\cdots<i_5$) such that
\begin{enumerate}
    \item For any two $r,\ell \in \{i_1,i_2,i_3,i_4,i_5\}$ with $\ell > r$, $|I_r \cap I_\ell| \le \epsilon |I_r|$, and
    \item $|I_{i_4}| \le (1+\alpha)\opt/n$,
\end{enumerate}
then $\obj(S,z) \le 1.999 \cdot \opt$. For that purpose, consider $T=\{x_{i_1}, x_{i_2},x_{i_3},x_{i_4},x_{i_5}\}$, and the output $\tilde{x}_T$ of {\BRO} on input $T$. We want to claim that $\obj(S,\tilde{x}_T) \le 1.999 \cdot \opt$, and hence $\obj(S,z) \le \obj(S,\tilde{x}_T) \le 1.999 \cdot \opt$, which will complete the analysis.

\begin{claim}
\label{clm:BRO-analysis}
For $T=\{x_{i_1}, x_{i_2},x_{i_3},x_{i_4},x_{i_5}\}$, $\obj(S,\tilde{x}_T) \le 1.999 \cdot \opt$.
\end{claim}
\begin{proof}
Let us define the set of \emph{bad} symbols as 
\[
B:=\cup_{r \ne \ell \in \{i_1,\cdots,i_5\}} (I_r \cap I_\ell).
\]
Note, 
\begin{equation}
    \label{eq:badsize}
    |B| \le 4 \epsilon |I_{i_1}| + 3\epsilon |I_{i_2}| + 2\epsilon |I_{i_3}| + \epsilon |I_{i_4}| \le 10\epsilon |I_{i_4}|
\end{equation}
(recall, by our assumption $|I_{i_1}| \le \cdots \le |I_{i_4}|$). Let us define the set of \emph{good} symbols as $G:=[d]\setminus B$. Observe, a symbol $a \in G$ if and only if $a \in I_r$ for at most one $r \in \{i_1,\cdots,i_5\}$. Hence, for any two distinct $a,b \in G$, for at least three $x \in T$, $a,b\not \in I_x$, in other words, both $a,b$ are aligned between $x$ and $x^*$. Thus, by the construction of $H$, $(a,b) \in E(H)$ if $a$ appears before $b$ in $x^*$, for every distinct $a,b \in G$.

Next, observe that for any subset $V \subseteq [d]$, for any vertex $v \in [d]$ and a shortest cycle $\mathcal{C}$ containing $v$ in the subgraph $H - V$,
\begin{enumerate}
    \item $\mathcal{C}$ must contain at least one bad symbol (i.e., from $B$);
    \item $\mathcal{C}$ must be of length 3.
\end{enumerate}
The first condition is straightforward since a set of good symbols cannot form a cycle (because they form a directed path according to their ordering in $x^*$). For the second condition, suppose $\mathcal{C}$ is of length strictly greater than 3 and $v,a,b$ are three consecutive vertices in $\mathcal{C}$. Observe, between any two vertices in the subgraph $H - V$, there is a directed edge (because for any two symbols $a_1,a_2$, either $a_1$ appears before $a_2$ or $a_2$ appears before $a_1$ in at least three permutations out of five). So either the edge $(b,v)$ or $(v,b)$ must be in the subgraph $H - V$. In the first case, we get a length 3 cycle consisting of $v,a,b$, and in the second case, we get a shorter cycle (by taking the edge $(v,b)$ while bypassing the vertex $a$ of $\mathcal{C}$) contradicting the fact that $\mathcal{C}$ is a shortest cycle containing $v$.

\iffalse
Next, observe that for any subset $V \subseteq [d]$, for a shortest cycle $\mathcal{C}$ in the subgraph $H - V$,
\begin{enumerate}
    \item $\mathcal{C}$ must contain at least one bad symbol (i.e., from $B$);
    \item $\mathcal{C}$ must contain at most two good symbols (i.e., from $G$).
\end{enumerate}
The first condition is straightforward since a set of good symbols cannot form a cycle (because they form a path according to their ordering in $x^*$). For the second condition, for the contradiction's sake, suppose there are at least three good symbols $a,b,c\in G$. Clearly, $a,b,c$ cannot appear consecutively; otherwise, the edge $(a,c) \in E(H)$ (since by the construction, it must hold that $a$ appears before $b$ in $x^*$ and $b$ appears before $c$ in $x^*$, and hence $a$ appears before $c$ in $x^*$), which creates a shorter cycle contradicting the fact that $\mathcal{C}$ is a shortest cycle. Let $P_1, P_2, P_3$ be the $a\to b, b\to c, c\to a$ path in $\mathcal{C}$. WLOG assume, $P_3$ contains at least three vertices.  By the construction, either the edge $(a,c) \in E(H)$ or $(c,a) \in E(H)$. In both cases, we get shorter cycles (for the first case, the cycle is formed by the edge $(a,c)$ and the path $P_3$, and for the second case, the cycle is formed by $P_1$, then $P_2$, the edge $(c,a)$) contradicting the fact that $\mathcal{C}$ is a shortest cycle.
\fi

Due to the above observation, after iterative cycle removal in {\BRO}, we get a subgraph $H'$ with $|V(H') \cap G| \ge |G| - 2|B|=d-3|B|$. Since $\tilde{x}'$ is a topological ordering of the vertices in $V(H')$, the length of a longest common subsequence between $x^*$ and $\tilde{x}'$ must be
\[
|\lcs(\tilde{x}',x^*)| \ge |V(H') \cap G|\ge d-3|B|.
\]
Hence,
\begin{align}
    \label{eq:close}
    \dt(\tilde{x}_T ,x^*) &= d-\lcs(\tilde{x}_T,x^*) \nonumber\\
    &\le d-\lcs(\tilde{x}',x^*)\nonumber \\
    &\le 3|B| \le 30\epsilon |I_{i_4}| &&\text{(by~\autoref{eq:badsize})}.
\end{align}
Recall, $|I_{i_4}| \le (1+\alpha)\opt/n$. So, we get that
\begin{align*}
    \obj(S,\tilde{x}_T)=\sum_{x_i \in S}\dt(x_i ,\tilde{x}_T)&\le \sum_{x_i \in S} (\dt(x_i,x^*)+\dt(x^*,\tilde{x}_T))&&\text{(by triangle inequality)}\\
    &\le (1+30\epsilon(1+\alpha))\opt\\
    &\le 1.999\cdot \opt
\end{align*}
for the choice of $\alpha=\epsilon/11=0.03319/11$.
\end{proof}

\subsection{Extension to the $k$-Median}
\label{sec:k-median}
Now we argue that our algorithm framework described so far can be extended to the $k$-median problem. More specifically, we show the following result.
\begin{theorem}
\label{thm:k-median}
There is a deterministic algorithm, that given a set $S$ of $n$ permutations over $[d]$, finds a $1.999$-approximate $k$-median in time $n^{O(k)}d \log d + O(n^5 d^3)$.
\end{theorem}
We would like to highlight that the above running time can further be improved to $(k \log (nd))^{O(k)}nd^3$ by running the algorithm described in this section on a sample set. However, such a modification slightly worsens the approximation factor. We describe the sampling procedure in detail in Section~\ref{sec:kstreaming}.

\begin{proof}
Here we briefly describe how to extend our median algorithm to the $k$-median problem. We build a set $M$ by adding the output permutations of the procedure {\BRO}($T$) for all the subsets $T\subseteq S$ such that $|T|=5$. Then, we also add the permutations in the input set $S$ to $M$. Finally, we output a subset $\tilde{Y}$ of $M$, of size at most $k$ that minimizes the $k$-median objective value, i.e., $\tilde{Y}=\argmin_{Y \subseteq M: |Y|\le k} \obj(S,Y)$. We refer to this algorithm as {\FkA}.

By the construction, $|M|=O(n^5)$, and by the running time analysis of the procedure {\BRO}, constructing the set $M$ takes time $O(n^5 d^3)$. The final step that outputs a subset of $M$ which minimizes the objective function, takes $n^{5k+1}d \log d$ time. So the overall running time is $n^{O(k)}d \log d + n^5 d^3$.

To argue about the approximation guarantee of the above algorithm, let us first consider an arbitrary optimal $k$-median $Y^*$ (which is of size at most $k$). This set $Y^*$ implicitly induces a partitioning of $S$ into at most $k$ clusters $C_1,C_2,\cdots,C_k$ (where some of the $C_i$'s could be empty depending on the size of $Y^*$). WLOG assume, $|Y^*|=k$ and $Y^*=\{y^*_1,\cdots,y^*_k\}$. Then $C_i:=\{x \in S \mid \dt(x,y^*_i) \le \dt(x,y^*_j)\text{ for all }j \ne i\}$ (if for some $x\in S$, $\dt(x,y^*_i)=\dt(x,y^*_j)$ for two $y^*_i \ne y^*_j$, then break the ties arbitrarily to form $C_i$'s). It is straightforward to see that $y^*_i$ is an optimal median of the set/cluster $C_i$. Then, by the analysis of {\FA} (in the previous section), a permutation $\tilde{y}_i$ will be added in $M$ that is a 1.999-approximate median of the cluster $C_i$, for all $i \in [k]$. Since in the final step we output a subset $\tilde{Y}$ of $M$, of size at most $k$ that minimizes the $k$-median objective value, clearly it would be a 1.999-approximate $k$-median of the input set $S$.
\end{proof}

\subsection{Extension to the $k$-median with outliers}
\label{sec:outlier}
The algorithm described in the previous section can produce a 1.999-approximate median even in the presence of outliers. In the $k$-median with outliers problem, we are given a parameter $p \in [0,1)$. Given an input set $S$, the problem then asks to find a set of size at most $k$ (which is not necessarily a subset of $S$) that minimizes the $k$-median objective value of a subset of $S$ of size at least $(1-p)|S|$. Formally, we define the objective value of $S$ with respect to a set $Y\subseteq \sym_d$ as $\obj_p(S,Y):=\min_{S'\subseteq S: |S'|\ge (1-p)|S|} \obj(S',Y)$. The problem asks to output a set $Y^* \subseteq \sym_d$ that minimizes $\obj_p(S,Y)$. Note, $\obj_0$ is the same as the standard $k$-median objective function $\obj$ (as defined in Section~\ref{sec:prelims}).

\begin{theorem}
\label{thm:k-median-outlier}
There is a deterministic algorithm, that given a set $S$ of $n$ permutations over $[d]$ and $p \in [0,1)$, finds a $1.999$-approximate solution to the $k$-median with outliers problem with parameter $p$, in time $n^{O(k)}d \log d + n^5 d^3$.
\end{theorem}
\begin{proof}
The algorithm is the same as that without outliers, i.e., that described in~\autoref{thm:k-median}, with the only exception that now we consider $\obj_p$ as the objective function. So the running time also remains the same. To argue about the approximation guarantee, let us first consider an arbitrary optimal $k$-median $Y^*$. Let the corresponding subset of $S$ be $S^*$ (i.e., $S^* = \argmin_{S'\subseteq S: |S'|\ge (1-p)|S|} \obj(S',Y^*)$). Then the argument would be exactly the same as that in~\autoref{thm:k-median} on the set $S^*$.
\end{proof}

\section{Streaming Algorithm for Approximating $k$-Median}
\label{sec:kstreaming}

In the streaming model, we are given a set $S$ of $n$ permutations $x_1, x_2, \dots, x_n$ over $[d]$ that arrive in a stream. 
Our objective is to design an algorithm that uses space $O(d\log^{20} n\log^6 d)$ and computes $k$ permutations $\tilde{y}_1, \dots, \tilde{y}_k$ over $[d]$ such that $\sum_{x\in S}\min(\Delta(\tilde{y}_1,x),\dots, \Delta(\tilde{y}_k,x))\le (2-\delta)\opt$ for some constant $\delta>0$ in polynomial time. Here $\opt$ denotes the optimal objective.

\mainthm*

\iffalse
\begin{theorem}
\label{thm:median-stream}
There is a randomized streaming algorithm that, given a set $S$ of $n$ permutations over $[d]$ arriving on a stream, with high probability, finds a $1.9999995$-approximate $k$-median using $k^2 d \polylog(nd))$ bits of space. The algorithm has update time $O(k^2d\log^2(kn) \log^2 d)$ and query time $(k\log (nd))^{O(k)}d^3$
\end{theorem}
\fi

Before proving the above theorem, let us first introduce a few tools which will be critical for our algorithm.

\paragraph{Coreset and streaming.}One of the important tools to solve the clustering problem is \emph{coresets}.
\begin{definition}[$(k,\epsilon)$-coreset]
\label{def:coreset}
For a set $S$ of points in an arbitrary metric space $\mathcal{X}$ and an implicit set $X\subseteq \mathcal{X}$ (of potential centers/medians), a weighted subset $P \subseteq S$ (with a weight function $w:P\to\mathbb{R}$) is a $(k,\epsilon)$-coreset of $S$ with respect to $X$ for the $k$-median problem if
\[
(1-\epsilon)\obj(S,Y) \le \sum_{x\in P}w(x)\cdot \dt(x,Y) \le (1+\epsilon)\obj(S,Y)
\]
for all subsets $Y\subseteq X$ of size at most $k$.
\end{definition}

There are several coreset constructions known in the literature. In this paper, we consider the following coreset construction, which is implied from~\cite{feldman2011unified} (and further explained in~\cite{BachemLL18, BJKW21}).
\begin{theorem}[\cite{feldman2011unified, BachemLL18, BJKW21}]
\label{thm:coreset}
There is an algorithm that, given a set $S$ of points of an arbitrary metric space $\mathcal{X}$ and an implicit set $X\subseteq \mathcal{X}$ (WLOG assume $S \subseteq X$), outputs a $(k,\epsilon)$-coreset of $S$ with respect to $X$ for the $k$-median problem, of size $O(\epsilon^{-2} k^2  \log |X|)$.
\end{theorem}

In this paper, we are interested in solving the $k$-median problem over the Ulam metric in the streaming model. We consider the \emph{insertion-only} streaming model, where a set of points (in our case, permutations) $x_1,x_2,\cdots,x_n$ arrive one after another in a streaming fashion. By combining~\autoref{thm:coreset} and the framework provided by~\cite{BravermanFLSZ21}, it is possible to build a coreset for the $k$-median problem over an arbitrary metric space using polylogarithmic space. More specifically, we use the following result.
\begin{theorem}
\label{thm:coreset-streaming}
There is a streaming algorithm that, given a set $S$ of points of an arbitrary metric space $\mathcal{X}$, arriving in an insertion-only stream and an implicit set $X\subseteq \mathcal{X}$ (WLOG assume $S \subseteq X$), maintains a $(k,\epsilon)$-coreset of the input with respect to $X$ for the $k$-median problem, by storing at most $O(\epsilon^{-2} k^2  \log |X| \log n )$ points of $S$. Furthermore, the algorithm has worst-case update time of $(\epsilon^{-1} k \log n)^{O(1)}$.
\end{theorem}

\paragraph{Monotone Faraway Sampling (MFS).}Another important tool that we will use to design a streaming algorithm for the $k$-median problem is the \emph{monotone faraway sampling (MFS)}, introduced in~\cite{braverman2021metric}. This sampling method allows us to sample "a few" points from an (insertion-only) stream such that the sample set includes a set of candidate medians that achieves $O(1)$-approximation. Although the approximation factor involved is much larger than 2, roughly speaking, it is sufficient for the clusters that contribute a small amount to the overall objective. We use the following result implied from~\cite{braverman2021metric}.

\begin{theorem}[\cite{braverman2021metric}]
    \label{thm:MFS}
    There is a streaming algorithm that, given a set $S$ of points of an arbitrary metric space $\mathcal{X}$, arriving in an insertion-only stream and parameters $\kappa,\rho \in (0,1)$, samples a subset $F\subseteq S$ of size $O(k^2 (\rho \kappa)^{-1}\log k \log (1+k\kappa n))$ such that the following holds: Suppose $Y^*=\{y^*_1,\cdots,y^*_k\}$ be an arbitrary optimum $k$-median of $S$ (where $\opt=\obj(S,Y^*)$) and let $C_1,\cdots,C_k$ denote the induced clustering of $S$. Then for each $i \in [k]$, there exists a $y'_i \in F$ such that 
 \[
 \sum_{x \in C_i}\dt(x,y'_i) \le 2\Big(1+\frac{1}{1-\kappa}\Big)\sum_{x \in C_i}\dt(x,y^*_i) + \rho \frac{\opt}{k}.
 \]
    Moreover, the algorithm requires both space and update time of $O(k^2 (\rho \kappa)^{-1}\log k \log (1+k\kappa n))$.
\end{theorem}

All the above algorithms are randomized and err with probability at most $1/10$. Now we are ready to describe our streaming algorithm for the $k$-median problem.

\paragraph{Algorithm Description.}
 The algorithm is similar to the $k$-median algorithm described in Section~\ref{sec:k-median}, and we refer to it as {\FkAS}. However, because of the space limitation, instead of storing all the permutations, we run the algorithm on a sample set instead of the whole input. 
 Let us consider the following set $M$ of (implicit) potential $k$-medians: $M$ contains all the input permutations. Further, it also contains all the output of $\BRO(T)$ for all $T \subseteq S$ such that $|T|= 5$. 
 Our algorithm works in two phases. In the first step, it samples a set $R$ of $O(\log ^4n \log d)$ permutations from the stream $x_1, \dots, x_n$. 
 Additionally, it also constructs a coreset $(P, w)$ for the set $M$ of (implicit) $k$-medians on input $S$ in a streaming fashion. 
 Then we show using these sampled permutations and the coreset we can compute $k$ permutations $\tilde{y}_1,\dots, \tilde{y}_k$ such that $\sum_{x\in S}\min(\Delta(\tilde{y}_1,x),\dots, \Delta(\tilde{y}_k,x))\le (2-\delta)\opt$ for some constant $\delta>0$.

\vspace{4mm}
\noindent
\textbf{Step 1 (Sampling Algorithm):}
Given set $S$ and parameters $\beta, \gamma >0$ (the values of which are to be fixed later), we sample a set of permutations $R$ from $S$ as follows. 

For each $\ell\in \{1, (1+\gamma), (1+\gamma)^2, \dots, d\}$ and $p\in \{1, \frac{1}{(1+\gamma)}, \frac{1}{(1+\gamma)^2}, \dots, \frac{1}{n}\}$, we create a set $S_{\ell,p}\subseteq S$ as follows: 

\vspace{2mm}
\noindent
\textbf{Step i)} For each $x_i$, discard $x_i$ with probability $1-p$.

\vspace{2mm}
\noindent
\textbf{Step ii)} If $\forall x_j\in S_{\ell,p}$, $\Delta(x_i,x_j)\ge \beta \ell$, add $x_i$ to $S_{\ell,p}$.

\vspace{2mm}
\noindent
\textbf{Step iii)} If $|S_{\ell,p}|\ge k\log^3 n$, set $S_{\ell,p}=\emptyset$.

Then set $R=\bigcup_{\ell,p}S_{\ell,p}$.

\vspace{4mm}
\noindent
\textbf{Step 2 (Monotone Faraway Sampling):}
 Consider parameters $\kappa=1/3$, $\rho>0$. Given the input set $S$, we use the monotone faraway sampling (MFS) from~\autoref{thm:MFS} due to~\cite{braverman2021metric} to get a subset $F \subseteq S$ of size $O(k^2\log k \log (kn))$ such that the following holds: Suppose $Y^*=\{y^*_1,\cdots,y^*_k\}$ be an arbitrary optimum $k$-median of $S$ (where $\opt=\obj(S,Y^*)$) and let $C_1,\cdots,C_k$ denote the induced clustering (to be defined formally later) of $S$. Then for each $i \in [k]$, there exists a $y'_i \in F$ such that 
 \begin{equation}
     \label{eq:MFS}
     \sum_{x \in C_i}\dt(x,y'_i) \le 2\Big(1+\frac{1}{1-\kappa}\Big)\sum_{x \in C_i}\dt(x,y^*_i) + \rho \opt/k = 5 \sum_{x \in C_i}\dt(x,y^*_i) + \rho \opt/k.
 \end{equation}

%\dcnote{5/2 should be 5? and need to change $\kappa$}

\vspace{4mm}
\noindent
\textbf{Step 3 (CoreSet Construction):}
Given a set of input permutations $S$ and a parameter $\lambda>0$, define a set $M$ containing all input permutations from $S$. Further, it also contains the output of \BRO($T$) for all $T\subseteq S$ such that $|T|=5$.
Following~\autoref{thm:coreset-streaming}, construct a $(k,\lambda)$-coreset $(P,w)$ for $S$ with respect to the implicit set $M$ of potential medians.

\vspace{4mm}
\noindent
\textbf{Step 4 (Computing Approximate $k$-median):}

At the end of the stream, we use $R$, $F$ and $(P,w)$ to simulate {\FkA}. More
specifically, we run algorithm \BRO() on every subset $T \subseteq R$ of size five, and then add those outputs $\tilde{x}_T$ to a set
$\tilde{M}$. Next, add all the elements of $R,F$ to $\tilde{M}$. Finally, for each $k$-tuple $(y_1,\dots, y_k) \in \tilde{M}^k$, compute 
$\sum_{x\in P} w(x)\min (\Delta(x, y_1),\dots,\Delta(x, y_k))$ and output the $k$-tuple that attains the minimum value.

\paragraph{Analyzing the algorithm.}
For analysis purpose we fix $k$ optimal medians $y^*_1,y^*_2,\dots,y^*_k$. For $i\in [k]$, define $C_i=\{x\in S \mid \forall j\in [k], \Delta(x,y^*_i)\le \Delta(x,y^*_j)\}$. Let $O_i=\obj(C_i,y^*_i)$. Thus $\opt=O_1+O_2+\dots, O_k$. We show our sampling algorithm satisfies the following.

\begin{lemma}
\label{lemma:streamsamplemain}
Consider $\gamma=0.1$. For any constant $\zeta>0$ and every $i\in [k]$,  if $O_i\ge \frac{\zeta \opt}{k}$, then $\exists \ell\in \{1, (1+\gamma), (1+\gamma)^2, \dots, d\}$ and $p\in \{1, \frac{1}{(1+\gamma)}, \frac{1}{(1+\gamma)^2}, \dots, \frac{1}{n}\}$ such that $S_{\ell,p}$ satisfies at least one of the following with high probability.
\begin{enumerate}
    \item $S_{\ell,p}$ contains a permutation $y$ where $\obj(C_i,y)\le (1.999999+\beta)O_i$.
    \item $S_{\ell,p}$ contains a subset $T=\{x_1,x_2,x_3,x_4,x_5\}$ such that $\obj(C_i,\tilde{x}_T)\le (1.995+61\beta)O_i$
\end{enumerate}
\end{lemma}

\begin{proof}
We start with a few definitions. Let $d_i=\frac{O_i}{|C_i|}$. For some constant $\alpha>0$, let $C_i^{far}$ be the set of permutations in $C_i$ whose distance from $y^*_i$ is $\alpha$ fraction more than the average distance $d_i$. Formally we define,

\[C_i^{far}=\{x \mid x\in C_i; \Delta(y^*_i,x)> d_i+\alpha d_i\}.\]

Let $C_i^{close}$ be the set of permutations in $C_i$ whose distance from $y^*_i$ is $\alpha$ fraction less than the average distance $d_i$. Formally we define,

\[C_i^{close}=\{x \mid x\in C_i; \Delta(y^*_i,x)< d_i-\alpha d_i\}.\]

Lastly, let $C_i^{avg}$ be the set of permutations in $C_i$ whose distance from $y^*_i$ is roughly the average distance $d_i$. Formally we define,

\[C_i^{avg}=\{x \mid x\in C_i; d_i-\alpha d_i \le \Delta(y^*_i,x)\le d_i+\alpha d_i\}.\]

\vspace{2mm}
\noindent
\textbf{Case 1:}
First we consider the case where $|C_i^{close}|\ge \frac{|C_i\setminus C_i^{far}|}{\log n}$.
Note by definition $|C_i^{far}|< \frac{|O_i|}{(1+\alpha)d_i}=\frac{|C_i|}{(1+\alpha)}$. Thus $|C_i\setminus C_i^{far}|\ge \frac{\alpha |C_i|}{(1+\alpha)}$ and $|C_i^{close}| \ge \frac{\alpha |C_i|}{(1+\alpha)\log n}$.
We consider the set $S_{\ell,p}$ where $\ell \le d_i<(1+\gamma)\ell$ and $\frac{p}{(1+\gamma)}<\frac{2log ^2 n}{|C_i|} \le  p$. 
As $|C_i^{close}|\ge \frac{\alpha |C_i|}{(1+\alpha)\log n}$, and $p\ge \frac{2 \log^2 n}{|C_i|}$, using Chernoff bound, with high probability the sampling algorithm samples at least $\frac{\alpha \log n}{10(1+\alpha)}$ permutations from $C_i^{close}$ in Step (i). Let $y$ be such a permutation. If $y$ survives in Step (ii) then it satisfies $\obj(y,C_i) \le (2-\alpha)O_i$ by triangle inequality. Otherwise $S_{\ell,p}$ contains a permutation $z$ such that $\Delta(z,y)\le \beta \ell \le \beta d_i$ (as $\ell \le d_i$). Thus $\obj(y,C_i)\le (2-\alpha+\beta)O_i$. 
By setting $\alpha=.0005$, we get $\obj(y,C_i)\le (1.995+\beta)O_i$.
Next, we show $S_{\ell,p}$ is never modified in Step (iii) and thus $z\in S_{\ell,p}$.

Again with high probability in Step (i) we sample at most $10\log^2 n$ permutations from $C_i$. Lastly, we argue that for each other cluster $C_j\neq C_i$ the following holds. Among the permutation sampled from $C_j$, at most one permutation which is at distance $<\frac{\beta \ell}{2}$ from $y^*_j$ survives at Step (ii). Otherwise let there be two permutations $p_1,p_2\in C_j$ such that both $p_1,p_2$ are at distance $<\frac{\beta \ell}{2}$ from $y^*_j$ and both of them survive in Step (ii). However, by triangle inequality, their distance is $<\beta \ell$ and we get a contradiction. Thus all but at most one permutation from $C_j$ that survives in Step (ii) will be at a distance of at least $\beta \ell/2$ from $y^*_j$. Let $T_j\subseteq C_j$ be the set of permutations that are at distance $\ge \beta \ell/2$ from $y^*_j$. Thus 

\begin{align*}
    \sum_{\substack{j\in[k]\\j\neq i}}|T_j| &\le\frac{2}{\beta \ell}\sum_{\substack{j\in[k]\\j\neq i}}O_j\\
    &\le \frac{2\opt}{\beta \ell}\\
    &\le \frac{2kO_i}{\zeta \beta \ell} &&\text{(since } O_i\ge \frac{\zeta \opt}{k} \text{)}\\
    &< \frac{2k(1+\gamma)O_i}{\zeta \beta d_i} &&\text{(since } d_i< (1+\gamma)\ell \text{)}\\
    &\le \frac{4k|C_i|}{\zeta \beta} &&\text{(since } d_i=\frac{O_i}{|C_i|} \text{)}
\end{align*}

As $p<\frac{4 \log^2 n}{|C_i|}$ with high probability we sample at most $\frac{100k \log^2 n}{\zeta\beta}$ permutations from $S\setminus C_i$. Thus with high probability $|S_{\ell,p}|\le k\log^3 n$ and it is never modified in Step (iii). 

\textbf{Case 2:}
Now on, we assume $|C_i^{close}|< \frac{|C_i\setminus C_i^{far}|}{\log n}$. Again as $|C_i\setminus C_i^{far}|\ge \frac{\alpha |C_i|}{(1+\alpha)}$ we have $|C_i^{avg}|\ge \frac{(\log n-1)|C_i\setminus C_i^{far}|}{\log n} \ge \frac{\alpha (\log n-1) |C_i|}{(1+\alpha)\log n}$.
For the analysis purpose, we fix an optimal alignment between each $x\in C_i$ and $y^*_i$ and let $I_{x}$ be the set of symbols from $[d]$ that are unaligned in this optimal alignment.
For a permutation $x\in C_i^{avg}$ we define set $C(x)=\{z| z\in C_i^{avg}; |I_x\cap I_z| \ge \epsilon d_i \}$, where $\epsilon >0$ is a constant.
Let $C_i^{avg,dense}=\{x| x\in C_i^{avg}; |C(x)|\ge \frac{|C_i^{avg}|}{6}\}$.
Here we consider the case where $|C_i^{avg,dense}|\ge \frac{|C_i^{avg}|}{10}\ge \frac{\alpha(\log n-1)|C_i|}{10(1+\alpha)\log n}$. 

Again we consider the set $S_{\ell,p}$ where $\ell \le d_i<(1+\gamma)\ell$ and $\frac{p}{(1+\gamma)}<\frac{2log ^2 n}{|C_i|} \le  p$. 
As $|C_i^{avg,dense}|\ge \frac{\alpha(\log n-1)|C_i|}{10(1+\alpha)\log n}$, and $p\ge \frac{2 \log^2 n}{|C_i|}$, with high probability the sampling algorithm samples at least $\log n/10$ permutations from $C_i^{avg,dense}$ in Step (i). Let $y$ be such a permutation. 
If $y$ survives in Step (ii) then it satisfies $\obj(y,C_i) \le 2O_i+(2\alpha +\frac{2}{\log n}-\frac{\epsilon}{6})|C_i\setminus C_i^{far}|d_i$ by triangle inequality (as in Section~\ref{sec:median}). Otherwise $S_{\ell,p}$ contains a permutation $z$ such that $\Delta(z,y)\le \beta \ell \le \beta d_i$. Thus $\obj(y,C_i)\le (2+\beta)O_i+(2\alpha +\frac{2}{\log n}-\frac{\epsilon}{6})|C_i\setminus C_i^{far}|d_i\le (2+\beta+\frac{2\alpha^2}{(1+\alpha)}+\frac{2\alpha}{(1+\alpha)\log n}-\frac{\epsilon \alpha}{6(1+\alpha)})O_i$. By setting $\epsilon=.0333$ and $\alpha=.0005$ we get $\obj(y,C_i)\le (1.999999+\beta)O_i$.

Following a similar argument as Case 1, we show $S_{\ell,p}$ is never modified in Step (iii).

\vspace{2mm}
\noindent
\textbf{Case 3:}
We define set $C_i^{avg, sparse}=C_i^{avg}\setminus C_i^{avg, dense}$. Now we consider the case where $|C_i^{avg,dense}|< \frac{|C_i^{avg}|}{10}$.
Thus $|C_i^{avg,sparse}|\ge \frac{9|C_i^{avg}|}{10}\ge \frac{9\alpha(\log n-1)|C_i|}{10(1+\alpha)\log n}$.

Again we consider the set $S_{\ell,p}$ where $\ell \le d_i<(1+\gamma)\ell$ and $\frac{p}{(1+\gamma)}<\frac{2log ^2 n}{|C_i|} \le  p$. We argue $S_{\ell,p}$ contains a set of five permutations $T=\{x_1,\dots, x_5\}$ such that $\forall i,j \in [5], |I_{x_i}\cap I_{x_j}|\le (\epsilon+2\beta) d_i$. For this, we define five events $e_1, e_2, \dots, e_5$. Here $e_1$ is the event that at least one permutation is sampled from $C_i^{avg,sparse}$. Note as $|C_i^{avg,sparse}|\ge \frac{9\alpha(\log n-1)|C_i|}{10(1+\alpha)\log n}$ and $p\ge \frac{2log ^2 n}{|C_i|}$ with high probability we sample at least 
$10\log n$ permutations from $C_i^{avg,sparse}$. Let $x_1$ be such a permutation. Next given $e_1$, let $e_2$ be the event that at least one permutation is sampled from $C_i^{avg,sparse}\setminus C(x_1)$. As $|C(x_1)|<\frac{|C_i^{avg}|}{6}$, $|C_i^{avg,sparse}\setminus C(x_1)|\ge \frac{9|C_i^{avg}|}{10}-\frac{|C_i^{avg}|}{6}\ge \frac{11\alpha(\log n-1)|C_i|}{15(1+\alpha)\log n}$. Thus with high probability, we sample at least 
$10\log n$ permutations from $C_i^{avg,sparse}\setminus C(x_1)$. Let $x_2$ be such a sting. In a similar way given $e_1, e_2, \dots e_{j-1}$ (where $j\le 5$) let $e_j$ be the event that at least one permutation is sampled from $C_i^{avg,sparse}\setminus (C(x_1)\cup \dots \cup C(x_{j-1}))$. Again using a similar argument as before we can show $e_j$ is satisfied with high probability and let $x_j$ be a permutation sampled from $C_i^{avg,sparse}\setminus (C(x_1)\cup \dots \cup C(x_{j-1}))$. Thus all these five events are satisfied together with high probability.

Given $T=\{x_1,\dots, x_5\}$, we can construct a permutation $\tilde{x}_T$ using the \BRO() algorithm (Algorithm~\ref{alg:ptwo}), such that $\obj(C_i,\tilde{x}_T)\le (1+30(\epsilon +2\beta)(1+\alpha))O_i$. By setting $\epsilon=.0333$ and $\alpha=.0005$ we get $\obj(C_i,\tilde{x}_T)\le (1.9995+61\beta)O_i$.

Again, following a similar argument as Case 1, we show $S_{\ell,p}$ is never modified in Step (iii).
\end{proof}

\iffalse
Following Lemma~\ref{} we can make the following claim.

\begin{claim}
\label{claim:kcoreset}
For every $i\in [k]$, the sample set $F$ (drawn using MFS procedure of Step 2) contains a permutation $y'_i$ such that $\obj(C_i,y'_i)\le 5O_i$.
\end{claim}
\begin{proof}
By the MFS procedure of~\cite{braverman2021metric}, it follows that for each $i \in [k]$, there exists a $y'_i \in F$ such that 
\[
\sum_{x \in C_i}\dt(x,y'_i) \le 2(1+\frac{1}{1-\kappa})O_i + \rho \opt/k.
\]
Recall, $O_i=\sum_{x \in C_i}\dt(x,y^*_i)$.
\end{proof}
\fi

%\begin{theorem}
%\label{claim:stremmain}
%For $\beta=0.0000001$, $\lambda=0.0000001$ and $\rho=0.00000001$ Algorithm Approx-$k$-Stream() outputs $k$ permutations $\tilde{y}_1,\dots,\tilde{y}_k$ such that $\sum_{x\in S}min(\Delta(\tilde{y}_1,x),\dots, \Delta(\tilde{y}_k,x))\le 1.9999994\opt$. Moreover the algorithm used $O(d\log^{20} n \log^6 d)$
%space and runs in polynomial time.
%\end{theorem}

\begin{proof}[Proof of~\autoref{thm:main}]
Given the input set of permutations $S$ arriving in a stream, we run Algorithm {\FkAS} with the following parameter setting: $\beta=0.0000001$, $\lambda=0.0000001$ and $\rho=0.00000001$. Let the algorithm outputs $k$ permutations $\tilde{y}_1,\dots,\tilde{y}_k$. We show with high probability $\sum_{x\in S}min(\Delta(\tilde{y}_1,x),\dots, \Delta(\tilde{y}_k,x))\le 1.9999995\opt$, the algorithm uses $k^2 d \polylog(nd))$ bits space and has update time $(k\log n)^{O(1)} d \log^2 d$  and query time $(k\log (nd))^{O(k)}d^3$.

%update time $O(k^2d\log^2(kn) \log^2 d)$

First, we argue the space complexity. In the sampling step (Step 1), there are $O(\log d)$ different choices for $\ell$ and $O(\log n)$ different choices for $p$. Moreover $|S_{\ell,p}|<k\log^3 n$. Thus the sampled set $R$ contains $O(k\log^4 n\log d)$ permutations. Moreover as we consider all subsets $T\subseteq R$ of size 5, there are at most $|R|^5=O(k^5\log^{20} n\log^5 d)$ different choices for $\tilde{x}_T$. However, instead of storing $\tilde{x}_T$ explicitly, we compute it whenever we use $\tilde{x}_T$ as a candidate $k$-median. Thus storing only set $R$ is enough.
Next, by the MFS algorithm of~\autoref{thm:MFS}, the sample set $F$ contains at most $O(k^2 \log k \log (kn))$ input permutations with high probability.% at least $1-1/n^2$.
Also, as the size of the implicit set $M$ can be bounded by $O(n^5)$, using~\autoref{thm:coreset-streaming}, we can claim that the size of the coreset $(P,w)$ is $O(k^2 \log^2 n)$.
Each permutation can be stored using $d\log d$ bits. Thus the total space is bounded by $k^2 d \polylog(nd))$.

Next, we show the algorithm has update time $O(k^2d\log^2(kn) \log^2 d)$ and query time $(k\log (nd))^{O(k)}d^3$.
For a given pair of permutations of length $d$, their distance can be computed in time $O(d\log d)$. As $|R|=O(k\log^4 n\log d)$, the update time of the first sampling step is $O(kd\log^4 n\log^2 d)$. 
Next, by the MFS algorithm of~\autoref{thm:MFS}, the update time to construct set $F$ and following~\autoref{thm:coreset-streaming} the update time to construct the coreset $(P,w)$ is $(k \log n)^{O(1)}$. Thus the total update time is $(k\log n)^{O(1)} d \log^2 d$.
Next, we argue the query time. As $|R|=O(k\log^4 n\log d)$, there are at most $|R|^5=O(k^5\log^{20}n \log^5 d)$ different choices for $\tilde{x}_T$. Thus $|\tilde{M}|=k^{O(1)}\polylog(nd)$. Now to bound the space complexity instead of storing $\tilde{x}_T$ explicitly, we compute it whenever it is used as a candidate $k$-median. 
Using Algorithm \BRO() it can be constructed in time $\tilde{O}(d^3)$. Moreover $|\tilde{M}|=k^{O(1)}\polylog(nd)$, total number of different candidate $k$-median is $(k\log (nd))^{O(k)}$. Also evaluating the total objective of a candidate $k$-median using the coreset takes time $O(k^3d\log(kn)\log d)$. Thus the total query time can be bounded by $(k\log (nd))^{O(k)}d^3$.

Lastly, we argue the approximation guarantee. Let $S_1 : = \{i\in [k] \mid O_i<\frac{\opt}{40000000k}\}$. %$S_1\subseteq [k]$ be the set of indices such that $i\in S_1$ iff $O_i<\frac{\opt}{40000000k}$. 
Thus $\sum_{i\in S_1}O_i\le \frac{\opt}{40000000}$. By~\autoref{eq:MFS}, for each $i\in [k]$, there exists $y'_i \in F$ such that 
\[
\obj(C_i,y'_i) \le 5 O_i + \frac{\rho}{k} \opt.
\]
Thus
\begin{align*}
    \sum_{i \in S_1} \obj(C_i,y'_i) &\le 5\sum_{i \in S_1} O_i + \rho \opt &&\text{(since $|S_1| \le k$)}\\
    &\le \frac{\opt}{10000000} &&\text{(for $\rho = 0.00000001$)}.
\end{align*}

Next following~\autoref{lemma:streamsamplemain} for each $i\in [k]\setminus S_1$, with high probability, either the sample set $R$ and thus $\tilde{M}$ contains a permutation $y$ such that $\obj(C_i,y)\le 1.9999991 O_i$ (as $\beta=0.0000001$) or $R$ contains a tuple $T=(x_1,\dots,x_5)$ such that $\obj(C_i,\tilde{x}_T)\le 1.9950061 O_i$. In this case in Step (iii) we compute $\tilde{x}_T$ and add this to $\tilde{M}$. Thus $\tilde{M}$ contains $k$ permutations with total objective $\le 1.999999\opt+\frac{\opt}{10000000}=1.9999992 \opt$. As we evaluate this using the $(k,\lambda)$ coreset $(P,w)$, the total objective is $1.9999992(1+\lambda) \opt\le 1.9999995\opt$ as $\lambda=0.0000001$.
\end{proof}

\section{Improved Space Bound for (1-)Median}
\label{sec:streaming-1-median}
In this section, we extend the result of Section~\ref{sec:median} in the streaming model. More specifically, suppose a set of input permutations $x_1,\cdots,x_n \in \sym_d$ arrive one after another as an insertion-only stream (a permutation, once arrived, cannot be deleted). Here we show a result similar to~\autoref{thm:median} in the streaming model using $o(nd)$ bits (note, the input size is $O(nd \log d)$ bits) of space.
\begin{theorem}
\label{thm:1-median-stream}
There is a randomized streaming algorithm that, given a set $S$ of $n$ permutations over $[d]$ arriving on a stream, finds a $1.9999$-approximate median using $O(d\log d \log^2 n)$ bits of space.
\end{theorem}
It is worth emphasizing that the space is only needed to maintain at most $O(\log n)$ permutations at any point of time of the stream.

\paragraph{Description of the algorithm.}The algorithm is essentially the same as that described in Section~\ref{sec:median}, except now we will run the algorithm on a sample set instead of the whole input. Let us consider the following (implicit) set $M$ (of potential medians): $M$ contains all the input permutations. Further, it also contains the output of {\BRO}($T$) for all $T\subseteq S$ such that $|T|=5$. Our algorithm consists of two components. First, it picks each input permutation independently with probability $\log n/n$. (Alternatively, we can use standard reservoir sampling to sample $O(\log n)$ permutations uniformly at random.) Let $R$ denote the sampled set. Second, consider an $\epsilon \in(0,1)$ (the value of which is to be fixed later). Then it constructs a $(1,\epsilon)$-coreset $(P,w)$ (where $P\subseteq S$ and $w:P\to \mathbb{R}$) for $S$  with respect to the (implicit) set $M$, in streaming fashion. At the end of the stream, we use $R$ and $(P,w)$ to simulate {\FA}. More specifically, we run {\BRO} on every subset $T\subseteq R$ of size five and then add those outputs to a set $M'$. Next, add all the elements of $R$ to $M'$. So,
\[
M'=R \cup \{\tilde{x}_T \mid \text{ for all }T\subseteq R \text{ such that }|T|=5\}
\]
where $\tilde{x}_T$ denotes the output of {\BRO}($T$). Finally, for each $y\in M'$, compute $\sum_{x \in P}w(x)\dt(x,y)$, and output one that attains the minimum value.

Since the analysis of the approximation factor is similar to that for the $k$-median problem (as in Section~\ref{sec:kstreaming}), we omit the details. The only difference is that since now we have only one cluster, we can entirely avoid the "clever" sampling used in Step 1 and the MFS sampling (Step 2) of the algorithm described in Section~\ref{sec:kstreaming}. Instead, we can just use the uniform sampling from the input and then follow a similar argument as used in~\autoref{lemma:streamsamplemain}. We focus on analyzing the space usage of our algorithm.

\paragraph{Space usage.}Clearly, the first sampling step requires $O(\log n)$ input permutations to store (with high probability). By~\autoref{thm:coreset-streaming}, the coreset construction needs to maintain at most $O(\epsilon^{-2} \log |M| \log n)=O(\epsilon^{-2} \log^2 n)$ (since $|M|=O(n^5)$) permutations. So the overall space usage is $O(\epsilon^{-2} d \log d \log^2 n )$ bits of space (since to store each permutation over $[d]$, we need $O(d \log d)$ bits of space).

\ifprocs
\bibliographystyle{plainurl}
\else
\bibliographystyle{alphaurl}
\fi
\bibliography{reference}

\end{document}